\documentclass[journal]{IEEEtran}

\usepackage{cite,url}
\usepackage[pdftex]{graphicx}
\usepackage[caption=false, font=footnotesize]{subfig}
\usepackage{amsmath, amssymb, amsthm}
\usepackage{algorithmic,algorithm}

\newtheorem{theorem}{Theorem}

\newtheorem{assumption}{Assumption}

\newtheorem{remark}{Remark}

\usepackage{tikz}
\usetikzlibrary{patterns}
\usetikzlibrary{shapes}

\begin{document}

\title{On Complexity Bounds for the Maximal Admissible Set of Linear Time-Invariant Systems}

\author{Hamid~R.~Ossareh,~\IEEEmembership{Senior Member,~IEEE,}
        Ilya~Kolmanovsky,~\IEEEmembership{Fellow,~IEEE}
\thanks{H.~R.~Ossareh is with the Department of Electrical and Biomedical Engineering, University of Vermont, Burlington, VT USA 05405. I.~Kolmanovsky is with the Department of Aerospace Engineering, University of Michigan, Ann Arbor, MI USA 48105.
e-mails: hossareh@uvm.edu, ilya@umich.edu}
}

\maketitle

\begin{abstract}
Given a dynamical system with constrained outputs, the maximal admissible set (MAS) is defined as the set of all initial conditions such that the output constraints are satisfied for all time. 
It has been previously shown that for discrete-time, linear, time-invariant, stable, observable systems with polytopic constraints, this set is a polytope described by a finite number of inequalities (i.e., has finite complexity). However, it is not possible to know the number of inequalities {\it apriori} from problem data. To address this gap, this contribution presents two computationally efficient methods to obtain {\it upper bounds} on the complexity of the MAS. The first method is algebraic and is based on matrix power series, while the second is geometric and is based on Lyapunov analysis. The two methods are rigorously introduced, a detailed numerical comparison between the two is provided, and an extension to systems with constant inputs is presented. Knowledge of such upper bounds can speed up the computation of MAS, and can be beneficial for defining the memory and computational requirements for storing and processing the MAS, as well as the control algorithms that leverage the MAS. 
\end{abstract}


\IEEEpeerreviewmaketitle

\section{Introduction}\label{sec:intro}
Consider a discrete-time linear time-invariant system
\begin{equation}\label{eq:unforced system}
\begin{aligned}
x(t+1) = Ax(t) \\
y(t) = Cx(t)
\end{aligned}
\end{equation}
where $t\in \mathbb{Z}^+$ is the discrete time index, $x(t)\in\mathbb{R}^n$ is the state vector, and $y(t) \in \mathbb{R}^q$ is the output vector. The output is required to satisfy the constraint
\begin{equation} \label{eq:constraint}
y(t) \in \mathbb{Y} 
\end{equation}
where $\mathbb{Y}$ is a compact polytope with the origin in its interior. 
This paper is concerned with the set of all initial conditions for which \eqref{eq:constraint} is satisfied for all time, that is:
\begin{equation}\label{eq:MAS unforced infinite}
O_\infty = \{x_0: C A^t x_0 \in \mathbb{Y}, \forall t\geq 0\}
\end{equation}
This set, which is referred to as the \textit{maximal admissible set} (MAS) \cite{Gilbert_1991}, is an invariant set that has been broadly employed in the control literature, for example, as a terminal constraint in the Model Predictive Control (MPC) optimization problem to guarantee closed-loop stability \cite{Camacho_2013,rawlings2017model}, or in Reference Governors and Command Governors to guarantee infinite-horizon constraint satisfaction \cite{Gilbert_1995_2, garone2017reference}. This set also plays a major role in the analysis of constrained systems and in set-theoretic methods in control, see e.g., \cite{blanchini1999set, blanchini2008set}. The properties and computations of this set, as well as its extensions to other classes of systems, have also received much attention in the literature, see e.g., \cite{kolmanovsky1998theory,Pluymers_2005, rakovic2008invariant, HIRATA2008526,ossareh2020reference, vaselnia2021inclusion,perez2011maximal,osorio2018stochastic,benfatah2021maximal}.

In the paper \cite{Gilbert_1991}, it was shown that if \eqref{eq:unforced system} is asymptotically stable and the pair $(A,C)$ is observable, then $O_\infty$ is a compact polytope which is finitely determined, i.e., it can be described by a finite number of inequalities:
\begin{equation}\label{eq:MAS unforced finite}
O_\infty = \{x_0: C A^t x_0 \in \mathbb{Y}, t = 0,\ldots,t^*\}
\end{equation}
where $t^*$, referred to as the \textit{admissibility index} of MAS, is the last ``prediction time-step" required to fully characterize the MAS. One difficulty, which the current paper seeks to overcome, is that $t^*$ is not known \textit{apriori} from problem data. To find it, one would need to construct the MAS iteratively by adding inequalities one time-step at a time and checking for redundancy of the newly added constraints. Once all the newly added constraints are redundant, $t^*$ has been found. To carry out the redundancy check,  Linear Programs (LPs) must be solved, which renders the construction of MAS computationally demanding for high dimensional systems, those with slow dynamics, those with constraint sets of high complexity, and in situations where $O_\infty$ must be computed in real-time, e.g., to accommodate changing models or constraints.

To fill this gap, this paper provides two methods to obtain an {\it upper bound} on $t^*$. This allows one to replace $t^*$ in \eqref{eq:MAS unforced finite} by its upper bound, thereby eliminating the need to solve LPs during the construction of MAS (at the expense of having potentially redundant inequalities in the set description). In addition to speeding up the computation of MAS, knowledge of such an upper bound is helpful for defining the memory and processing requirements to store and employ the MAS for the purpose of control. Moreover, from a theoretical standpoint, the two methods presented here can be viewed as alternative justifications for the finite determinism of the MAS and, more specifically, the existence of $t^*$.

The first method for finding an upper bound on $t^*$ is algebraic, and leverages matrix power series to express the output at a time $t$ as a linear combination of outputs at previous times, which helps us determine the time-step after which the constraints  become redundant. The second method is geometric and relies on the decay rate of a quadratic Lyapunov function towards a constraint-admissible ellipsoidal level set. Both methods are computationally efficient and do not rely on optimization solvers. To the best of our knowledge, the first method is new. The second method is inspired by the existing literature (see e.g., \cite{garone2017reference,CDCworkshop}); however, it is presented here in complete details with explicit bounds, and several enhancements to it are proposed.

This paper presents the theoretical justification for both methods, as well as corresponding algorithms for the computation of the upper bounds. The upper bounds obtained from the two methods are then compared against the true value of $t^*$ using a Monte Carlo study. It is shown that the first method results in a tighter upper bound as compared with the second method for all the random systems considered. The upper bounds also provide insight into the fundamental nature of $t^*$ itself. Specifically, it is shown that $t^*$ is closely related to the spectral radius of $A$.

Finally, the two methods are extended to systems with constant inputs, which have been studied extensively in the literature on reference and command governors: 
\begin{equation}\label{eq:forced system}
\begin{aligned}
x(t+1) = Ax(t) + Bu\\
y(t) = Cx(t) + Du
\end{aligned}
\end{equation}
where $u \in \mathbb{R}^m$ is a constant input. The definition of MAS for \eqref{eq:forced system} is similar to \eqref{eq:MAS unforced infinite}, but modified to account for the input:
\begin{equation}\label{eq:MAS forced infinitee}
O_\infty = \{(x_0,u): y(t) \in \mathbb{Y}, \forall t\geq 0\}
\end{equation}
It is shown in \cite{Gilbert_1991} that this set is generally {\it not} finitely determined (i.e., it cannot be described by a finite number of inequalities). However, a finitely-determined, positively-invariant inner approximation, denoted by $\widetilde{O}_\infty$, can be obtained by tightening the steady-state constraint: 
\begin{equation}\label{eq:MAS forced}
\begin{aligned}
\widetilde{O}_\infty = \{(x_0,u): y(\infty) \in (1-\epsilon) \mathbb{Y}, y(t) \in \mathbb{Y},
t = 0,\ldots,t^*\}
\end{aligned}
\end{equation}
where $\epsilon \in (0,1)$ is typically a small number. In \eqref{eq:MAS forced}, $y(\infty)=H_0 u$, where $H_0$ is the DC gain. Similar to the unforced case, the admissibility index, $t^*$, for this case is not known {\it apriori} from problem data. We thus extend the two methods described previously to find upper bounds on $t^*$. As we show, the upper bounds depend explicitly on the value of $\epsilon$. A Monte Carlo study similar to the one described above is conducted to compare the two methods. Similar to the case of unforced systems, Method 1 results in tighter upper bounds for all random systems considered.

The outline of this paper is as follows. Section~\ref{sec:unforced} presents the two methods described above for the unforced case and provides a numerical study to compare them. Section~\ref{sec:forced} extends the results to the case of systems with constant inputs. Conclusions and future works are provided in Section~\ref{sec:conclusions}. 

The notation in this paper is as follows: $\mathbb{Z}^+$, $\mathbb{R}$, $\mathbb{R}^n$, $\mathbb{R}^{n\times n}$, and $\mathbb{C}$ denote the set of non-negative integers, real numbers, $n$-dimensional vectors of real numbers, $n\times n$ matrices with real entries, and complex numbers, respectively.  For a symmetric matrix $P=P^{\sf T}$, we say it is positive definite and write $P \succ 0$ if all the eigenvalues of $P$ are strictly positive. We use the variables $t \in \mathbb{Z}^+$, $t^* \in \mathbb{Z}^+$, and $m \in \mathbb{Z}^+$ to denote the discrete time index, the admissibility index of MAS, and the upper bound on the admissibility index, respectively.

\section{Main Results: Unforced Systems}\label{sec:unforced}

Consider system \eqref{eq:unforced system} with constraint \eqref{eq:constraint}. Our goal is to obtain an upper bound on $t^*$ in \eqref{eq:MAS unforced finite}. This section presents two methods to obtain such an upper bound. The first method, which we refer to as ``Method 1", is based on a matrix power series expansion and the second, which we refer to as ``Method 2", is based on Lyapunov analysis. Consistent with the assumptions in the MAS literature (see, e.g., \cite{Gilbert_1991}), we assume that:
\begin{assumption}
System \eqref{eq:unforced system} is asymptotically stable and the pair $(A,C)$ is observable. Furthermore, the constraint set in \eqref{eq:constraint} is described by 
\begin{equation}\label{eq:constraint detailed}
\mathbb{Y} := \{y: -y_j^l \leq y_j \leq y_j^u, \quad j = 1,\ldots,q \}
\end{equation}
where $y_j$ is the $j$-th element of $y$, $y_j^l>0$ defines the lower limit for $y_j$, and $y_j^u>0$ defines the upper limit for $y_j$.
\label{assumption1}
\end{assumption}

\subsection{Method 1: Matrix Power Series}

The general idea behind this method is to first expand $A^t$ in terms of lower powers of $A$. This expansion allows us to express the output $y(t)$ in \eqref{eq:unforced system} as a linear combination of the outputs at previous times. We show that if there exists an integer $m$ such that the  sum of the coefficients in the expansion of $A^{m+1}$ is ``sufficiently small", then $m$ is an upper bound on $t^*$. We then show that such an expansion always exists thanks to the Cayley Hamilton Theorem.





We begin by stating the main result of this section.

\begin{theorem}\label{theorem:unforced}
Consider system \eqref{eq:unforced system} with constraint \eqref{eq:constraint}, and  
suppose Assumption \ref{assumption1} holds. Suppose there exists an integer $m$, $m \geq 0$, such that $A^{m+1}$ can be expanded as:
\begin{equation}\label{eq:expansion}
A^{m+1} = \sum\limits_{i=0}^{m} \alpha_i A^{i}
\end{equation}
where $\alpha_i\in\mathbb{R}$, $i=0,\ldots,m$, satisfy the following 
condition:
\begin{equation}\label{eq:condition on alphai}
\sum_{\alpha_i > 0} \alpha_i - \gamma 
\sum_{\alpha_i < 0} \alpha_i  \leq 1
\end{equation}
where $\gamma$ is the largest asymmetry in the constraints, i.e., 
\begin{equation}\label{eq:gamma}
\gamma = \max_j \left\{\max \Big\{\frac{y_j^u}{y_j^l}, \frac{y_j^l}{y_j^u}\Big\}\right\}
\end{equation}
Then, $m$ is an upper bound on the admissibility index, $t^*$, of the MAS for \eqref{eq:unforced system}--\eqref{eq:constraint}; that is, $t^* \leq m$.
\end{theorem}
\begin{proof}
To show that $m$, as defined in the theorem, is an upper bound on $t^*$, we must prove that $y(t)\in \mathbb{Y}$ for $t \leq m$ implies that $y(t)\in \mathbb{Y}$ for all $t \geq m+1$, i.e., the latter inequalities are implied by the former and, hence, redundant. We prove this assertion using mathematical induction. 

Induction base case: Assume $y(t)\in \mathbb{Y}$ for $t \leq m$ and show that $y(m+1) \in \mathbb{Y}$. To show this, note that the $j$-th output, starting from an initial condition, $x_0$, can be expanded using Eq.~\eqref{eq:expansion}:
$$
y_j(m+1) = C_j A^{m+1} x_0 = \sum\limits_{i=0}^{m} \alpha_i (C_j A^{i} x_0)
$$
The assumption $y(t) \in \mathbb{Y}$ for $t\leq m$ implies that $C_j A^ix_0$ in the above sum satisfies: $-y_j^l \leq C_j A^ix_0\leq y_j^u$. Thus, if $\alpha_i>0$, we have that $-\alpha_i 
 y_j^l \leq \alpha_i  C_j A^ix_0\leq \alpha_i y_j^u$, and if $\alpha_i < 0$, we have that $\alpha_i y_j^u \leq \alpha_i  C_j A^ix_0\leq -\alpha_i  y_j^l$. Thus, summation over $i$ results in:
 $$
 -y_j^l \sum_{\alpha_i>0}\alpha_i + y_j^u \sum_{\alpha_i<0}\alpha_i \leq  y_j(m+1) \leq y_j^u \sum_{\alpha_i>0}\alpha_i  - y_j^l \sum_{\alpha_i<0}\alpha_i
 $$
To ensure that $-y_j^l \leq  y_j(m+1) \leq y_j^u$, it suffices for $\alpha_i$ to satisfy:
\begin{equation}\label{eq:proof1}
 -y_j^l \leq -y_j^l \sum_{\alpha_i>0}\alpha_i + y_j^u \sum_{\alpha_i<0}\alpha_i
\end{equation}
\begin{equation}\label{eq:proof2}
 y_j^u \sum_{\alpha_i>0}\alpha_i  -y_j^l \sum_{\alpha_i<0}\alpha_i \leq  y_j^u 
\end{equation}
or if we divide both sides of \eqref{eq:proof2} by $y_j^u>0$, and both sides of \eqref{eq:proof1} by $-y_j^l<0$, it suffices that
$$
   \sum_{\alpha_i>0}\alpha_i - \frac{y_j^u}{y_j^l}\sum_{\alpha_i<0}\alpha_i \leq 1
$$
$$
    \sum_{\alpha_i>0}\alpha_i - \frac{y_j^l}{y_j^u}\sum_{\alpha_i<0}\alpha_i \leq 1
$$
Both of these inequalities hold as they are implied by \eqref{eq:condition on alphai}. Thus, $-y_j^l \leq y_j(m+1) \leq y_j^u$. Since $j$ was arbitrary, we have that $y(m+1) \in \mathbb{Y}$, as desired. 

Induction main step: Assume $y(t)\in \mathbb{Y}$ for $t \leq k$, where $k \geq m+1$, and show that $y(k+1) \in \mathbb{Y}$. We again write the $j$-th output:
$y_j(k+1) = C_j A^{k+1} x_0$ but now decompose $A^{k+1} = 
A^{m+1} A^{k-m}$. We thus obtain:
$$
y_j(k+1) = C_j A^{m+1} A^{k-m} x_0 = \sum\limits_{i=0}^{m} \alpha_i (C_j A^{i+k-m} x_0)
$$
The assumption $y(t) \in \mathbb{Y}$ for $t\leq k$ together with $0 \leq i+k-m \leq k$ imply that $C_j A^{i+k-m} x_0$ in the above sum satisfies: $-y_j^l \leq C_j A^{i+k-m}x_0\leq y_j^u$. The rest of the proof from this point on follows the same arguments as in the induction base case. This concludes the proof.
\end{proof}

\begin{remark}
In the case of symmetric constraints, the expression in Theorem 1 can be further simplified. Specifically, suppose  that $y_j^l = y_j^u, \forall j$ in \eqref{eq:constraint detailed}. Then, $\gamma = 1$ and so condition \eqref{eq:condition on alphai} becomes:
\begin{equation}\label{eq:unforced condition symmetric}
\sum_{i} |\alpha_i| \leq 1
\end{equation}
\end{remark}

\begin{remark}
If all the coefficients in the expansion of $A^{m+1}$ are positive, then \eqref{eq:condition on alphai} becomes:
$$
\sum_{i} \alpha_i \leq 1
$$
which is completely independent of the constraint set (i.e., independent of the $C$ matrix, $y_j^l$, and $y_j^u$).
\end{remark}

We now prove the existence of, and a develop a method to construct, the expansion in \eqref{eq:expansion} satisfying condition \eqref{eq:condition on alphai}. We first recall some facts from linear algebra. The characteristic polynomial of any square matrix $A \in \mathbb{R}^{n\times n}$ is defined as $\Delta(s) :=\det (sI-A)$. It is an $n$-th degree polynomial whose roots are the eigenvalues, $\lambda_i \in \mathbb{C}$, of $A$. We can thus write:
\begin{equation}\label{eq:characteristic polynomial}
\Delta (s) = (s-\lambda_1)\cdots(s-\lambda_n) = s^n + c_{n-1} s^{n-1} + \ldots + c_1 s + c_0
\end{equation}
The Cayley Hamilton theorem states that any square matrix satisfies its own characteristic polynomial:
\begin{theorem}[see \cite{Chen_1998}]
Let $A \in \mathbb{R}^{n\times n}$ be a matrix and let $\Delta (s)$ be its characteristic polynomial. Then, $\Delta(A)=0$.
\end{theorem}
This result allows us to express $A^t$, for any $t\geq n$, as a finite power series in lower powers of $A$. Specifically, $A^n$ can be expanded as: 
\begin{equation}\label{eq:CH n}
A^n = -c_0 I - c_1 A - \ldots - c_{n-1} A^{n-1} 
\end{equation}
where $c_i$ are the coefficients in \eqref{eq:characteristic polynomial} and are uniquely defined. 
Similarly, $A^{n+1}$ can be expanded in the same powers of $A$:
\begin{align*}
A^{n+1} = A(A^n) &= -c_0 A  - \ldots - c_{n-2} A^{n-1} - c_{n-1} A^{n} \\ 
&= (c_0 c_{n-1})I+(-c_0+c_1c_{n-1})A+\ldots+\\
&\quad\,\,(-c_{n-2}+c_{n-1}c_{n-1})A^{n-1} 
\end{align*}
Generalizing the above to any $t \geq n$, one can expand $A^t$ as:
\begin{equation}\label{eq:CH Expansion}
A^t = \sum\limits_{i=0}^{n-1} \beta_i(t)A^i
\end{equation}
where $\beta_i(t)$ denotes the $i$-th coefficient in the expansion of the $t$-th power of $A$. Note that expansion 
of $A^t$ in lower powers of $A$ is generally not unique, but $\beta_i(t)$ in \eqref{eq:CH Expansion} are, by construction, uniquely defined.

To simplify the presentation, we stack the coefficients of the $t$-th power into a vector and denote it by $\beta(t)$:
$$
\beta(t) = [\beta_0(t) \cdots \beta_{n-1}(t)]^T \in \mathbb{R}^n
$$
The following lemma characterizes $\beta(t)$ and its convergence properties as $t \rightarrow \infty$.
\begin{theorem}
Let $A \in \mathbb{R}^{n\times n}$ be any square matrix and let $\beta(t), t\geq n$, be the vector of coefficients in the expansion of $A^t$, as defined above. Then, $\beta(t)$ satisfies the difference equation
\begin{equation}\label{eq:CH recursion}
\beta(t+1) = \left[\begin{array}{c c c c c} 0 & 0 & \cdots & 0 & -c_0 \\ 1 & 0 &\cdots &0 &-c_1 \\ 0 & 1 &\cdots & 0 &-c_2 \\ \vdots & \vdots  && \vdots \\ 0 & 0 &\cdots &1 &-c_{n-1}  \end{array} \right] \beta(t)
\end{equation}
with initial condition $\beta(n)=[-c_0\,\,\, \cdots\,\,\, -c_{n-1}]^T$. In addition, if $A$ is asymptotically stable, then $\lim_{t\rightarrow \infty} \beta(t)=0$.
\label{lemma:CH coeff}
\end{theorem}
\begin{proof}
The initial condition is already shown in Eq. \eqref{eq:CH n}. To derive the recursion, suppose $A^{t}  = \sum_{i=0}^{n-1} \beta_i (t) A^i
$, where $\beta_i(t)$ are given. To find $\beta(t+1)$ in terms of $\beta(t)$, we expand $A^{t+1}$ as follows:
$$
A^{t+1} = A(A^t) = A\sum_{i=0}^{n-1} \beta_i (t) A^i = \sum_{i=0}^{n-1} \beta_i (t) A^{i+1}
$$
$$
= \beta_{n-1}(t)A^n + \sum_{i=0}^{n-2} \beta_i (t) A^{i+1} 
$$
$$
= \beta_{n-1}(t) \sum_{i=0}^{n-1} -c_i A^i + \sum_{i=0}^{n-2} \beta_i (t) A^{i+1} 
$$
$$= -\beta_{n-1}(t)c_0I + \sum_{i=1}^{n-1} \left(-\beta_{n-1}(t)c_i+\beta_{i-1}(t)\right) A^i
$$
Thus, $\beta_0(t+1) = -\beta_{n-1}(t)c_0$ and $\beta_i(t+1) = -\beta_{n-1}(t)c_i+\beta_{i-1}(t)$ for $i=1,\ldots,n-1$. This coincides with  recursion \eqref{eq:CH recursion}.

To prove that $\beta(t)\rightarrow 0$ as $t \rightarrow \infty$, note that the matrix in \eqref{eq:CH recursion} is exactly the  observable canonical form of matrix $A$, see \cite{Chen_1998} for details. Since this matrix can be obtained through a similarity transformation of $A$, it has the same eigenvalues as $A$. Thus, the recursion in \eqref{eq:CH recursion} corresponds to an asymptotically stable dynamical system and, hence, $\beta(t)$ must converge to 0. This proves the lemma.
\end{proof}

The above lemma guarantees the existence of an integer $m$ such that the coefficients of the expansion of $A^{m+1}$ as defined in \eqref{eq:expansion} satisfy condition \eqref{eq:condition on alphai}. To see this, compute $\beta(t)$ using the recursion in \eqref{eq:CH recursion} for increasing $t$ starting from $t=n$, and stop when
\begin{equation}\label{eq:condition on betai}
\sum_{\beta_i(t) > 0} \beta_i(t) - \gamma 
\sum_{\beta_i(t) < 0} \beta_i(t)  \leq 1
\end{equation}
Note that such $t$ always exists, because according to Theorem~\ref{lemma:CH coeff}, $\beta(t)\rightarrow 0$ as $t \rightarrow \infty$ and thus the left hand side of \eqref{eq:condition on betai} can be made arbitrarily small. Such $t$ corresponds to $m+1$ in Theorem \ref{theorem:unforced}, where the $\alpha_i$ in \eqref{eq:condition on alphai} are related to $\beta_i(t)$ in \eqref{eq:condition on betai} as follows: $\alpha_i = \beta_i(t)$ for $i=0,\ldots,n-1$ and $\alpha_i=0$ for $i=n,\ldots,m$. 
This leads to Algorithm \ref{Algorithm:algorithm1} for finding an upper bound for $t^*$.

\begin{algorithm}
 \caption{Compute upper bound on $t^*$ using Method 1}
     \hspace*{\algorithmicindent} \textbf{Input: $A, y_j^l, y_j^u$} \\
    \hspace*{\algorithmicindent} \textbf{Output: $m$ such that $t^*\leq m$} \\
 \begin{algorithmic}[1]
 \STATE Compute the Cayley Hamilton coefficients, $c_i$, using \eqref{eq:characteristic polynomial}, and $\gamma$ using \eqref{eq:gamma}.
 \STATE Set $t=n$ and initialize $\beta(n)$ as in Theorem \ref{lemma:CH coeff}.
 \STATE  If $\beta(t)$ satisfies \eqref{eq:condition on betai}, then: $m=t-1$, STOP.
 \STATE Increment $t$ by 1. Compute $\beta(t)$ using \eqref{eq:CH recursion}. Go to step 3.
 \end{algorithmic} 
 \label{Algorithm:algorithm1}
 \end{algorithm}

The above results allow us to say more about the value of $t^*$ itself in the case of first order systems. 
\begin{theorem}
Consider \eqref{eq:unforced system} with $A\in\mathbb{R}$ (i.e., a first order system) and assume Assumption \ref{assumption1} holds. If $A>-\frac{1}{\gamma}$, then $t^*=0$. In particular, if the constraints are symmetric (i.e., $\gamma = 1$ in \eqref{eq:gamma}), then $t^*=0$.
\end{theorem}
\begin{proof}
From \eqref{eq:CH n}, it follows that $A=-c_0$ and from Theorem~\ref{lemma:CH coeff}, $\beta(t) \in \mathbb{R}$ satisfies $\beta(1)=-c_0=A$. Asymptotic stability of $A$ and  condition $A > -\frac{1}{\gamma}$ imply that $-\frac{1}{\gamma}<\beta(1)<1$. Using this condition, it can be seen that, regardless of sign of $\beta(1)$,  \eqref{eq:condition on betai} is satisfied for $t=1$. Per Algorithm~\ref{Algorithm:algorithm1}, an upper bound on $t^*$ is therefore $m=0$, which implies that $t^*=0$.
\end{proof}
This theorem suggests that the MAS for some first order systems is particularly straightforward to construct.

We conclude this section with a few remarks. 

\begin{remark}
The Cayley Hamilton-based expansion in \eqref{eq:CH Expansion} provides only one possible expansion for $A^{m+1}$  in Theorem~\ref{theorem:unforced}. There may be other expansions that lead to smaller upper bounds for $t^*$.
\end{remark}

\begin{remark}
It has been shown in \cite{Gilbert_1991,freiheit2020overshoot} that $t^*$ does not change if the constraint limits $y_j^l$ and $y_j^u$ are multiplied by a scalar (i.e., the constraint set is radially scaled). From Algorithm \ref{Algorithm:algorithm1}, it can be seen that that this is also the case for the upper bound on $t^*$, because the ratio of $y_j^l$ and $y_j^u$ (through $\gamma$) is the only information about the constraint that is used by the algorithm. 
\end{remark}

\begin{remark}\label{remark:bound}
The Cayley Hamilton-based upper bound sheds light on the conditions under which $t^*$ may be large. Specifically, as the recursion in Theorem \ref{lemma:CH coeff} suggests, 
the upper bound on $t^*$ depends on the eigenvalues of $A$. If the spectral radius of $A$ is large (i.e., there is an eigenvalue close to the boundary of the unit disk), the upper bound 
 on $t^*$ (and likely $t^*$ itself) will be large. On the other hand, if the spectral radius is small, then the upper bound will be small (and thus $t^*$ must also be small). Thus, there is a relationship between $t^*$ and the spectral radius of $A$, which we examine numerically in Section \ref{sec:numerical comparison unforced}. An interesting implication of this is the following: if $A$ is obtained by discretizing a continuous-time model, the eigenvalues of $A$ approach the origin as the sampling period increases, leading to smaller values for the upper bound on $t^*$ and thus smaller values for $t^*$. Thus, there is also a relationship between $t^*$ and the sampling rate for the discretization.
\end{remark}

\subsection{Method 2: Lyapunov Level Sets}

The second method to find an upper bound on $t^*$ relies on Lyapunov level sets. We begin by defining two sets:
\begin{equation}\label{eq:X}
\mathbb{X}=\{x: Cx \in \mathbb{Y}\}
\end{equation}
which is the inverse image of $\mathbb{Y}$ in the $x$-space, and 
\begin{equation}\label{eq:On-1}
O_{n-1} = \{x_0: C A^t x_0 \in \mathbb{Y}, t=0,\ldots,n-1\}
\end{equation}
which is the set of all initial conditions such that the constraints are satisfied for the first $n$ time-steps. From these definitions, it follows that 
$$
O_\infty \subseteq O_{n-1} \subseteq \mathbb{X}
$$
The set $\mathbb{X}$ is not generally compact, but it is shown in \cite{Gilbert_1991} that, under Assumption \ref{assumption1}, $O_\infty$ and $O_{n-1}$ are. The compactness of $O_{n-1}$ is the main reason why it is employed in the analysis that follows. If $\mathbb{X}$  itself is compact, then $O_{n-1}$ may be replaced by $\mathbb{X}$ in the subsequent presentation.

Define the quadratic Lyapunov function
\begin{equation}\label{eq:lyapunov}
V(x) = x^TPx
\end{equation}
where $P=P^{\sf T} \succ 0$ is the solution of the discrete Lyapunov equation 
\begin{equation}\label{eq:lyap Q}
A^TPA-P=-Q
\end{equation}
for a given $Q=Q^{\sf T} \succ 0$. For each real number $r>0$, the $r$-th level set of $V(x)$, defined by
\begin{equation}\label{eq:lyapunov level sets}
\Omega_r = \{x: V(x) \leq r\},
\end{equation}
is an ellipsoid and is positively invariant 
with respect to the dynamics of \eqref{eq:unforced system}, see  \cite{Chen_1998,khalil2002nonlinear}. 

The key idea behind Method 2 is to first find two level sets of $V(x)$: one that is inscribed in $\mathbb{X}$ and one that circumscribes $O_{n-1}$. 
We determine the worst-case decay rate of $V(x)$ in time, and quantify the longest time it takes for the system state to enter the smaller level set (and thus satisfy the constraints for all future times) starting from anywhere in the larger level set (which is an outer approximation of $O_{n-1}$). This time provides an upper bound on $t^*$. See Fig. \ref{fig:lyapunov} for an illustration.

We now formally examine the above ideas, and then provide a method to find a suitable matrix $Q$ for the problem at hand.

\begin{figure}
\centering
\begin{tikzpicture}
    minimum height = 2cm] at (0,0) {};
\draw[pattern=north west lines, pattern color=blue] (-1,-1) rectangle (1,1);
    \draw[-] (-3,1) -- (3,1);
    \draw[-] (-3,-1) -- (3,-1);
     \draw [rotate=45](0,0) ellipse (1.17cm and 0.8cm);
      \draw [rotate=45](0,0) ellipse (2cm and 1.4cm);
    \draw[thick, ->] (0,-2) -- (0,2) node[above] {$x_2$};
    \draw[thick, ->] (-4,0) -- (4,0) node[right] {$x_1$};
\node at (2.6,0.6) 
    {$X$};
\node at (0.45,0.3) 
    {$O_{n-1}$};
\end{tikzpicture}
\caption{Figure illustrating the key idea behind the second method. The set $\mathbb{X}$ is defined in \eqref{eq:X} and is illustrated as the strip between the two horizontal lines. The set $O_{n-1}$ is defined in \eqref{eq:On-1} and is the hatched box. The smaller ellipse is the largest level set of the Lyapunov function $V(x)$ inscribed in $\mathbb{X}$. The larger ellipse is the smallest level set of $V(x)$ circumscribing $O_{n-1}$.}
\label{fig:lyapunov}
\end{figure}
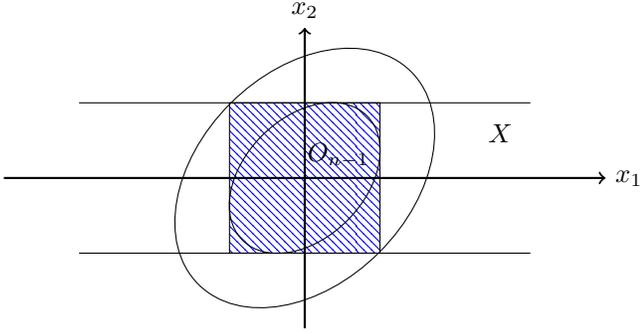

\begin{theorem}\label{theorem:lyapunov bound}
Consider system \eqref{eq:unforced system} with Lyapunov function \eqref{eq:lyapunov}--\eqref{eq:lyap Q} and constraint \eqref{eq:constraint}, and suppose Assumption \ref{assumption1} holds. Define $r_1, r_2 \in \mathbb{R}$ as follows: 
$$
r_1 = \max \{r:  \Omega_r \subset \mathbb{X}\}
$$
$$
r_2 = \min \{r: O_{n-1} \subset \Omega_r\}
$$ 
Then, we have that $r_2 \geq r_1$ and that an upper bound for $t^*$ is given by:
\begin{equation}\label{equ:m_method2}
m = \mathrm{floor}\bigg( \frac{\log (\frac{r_1}{r_2})}{\log( \sigma )}\bigg)
\end{equation}
where the {\normalfont floor} operator returns the previous largest integer,
\begin{equation}\label{eq:sigma}
\sigma = 1-\frac{\lambda_{min}(Q)}{\lambda_{max}(P)},
\end{equation}
and $\lambda_{min}$ and $\lambda_{max}$ denote the smallest and largest eigenvalues, respectively. 
\end{theorem}

\begin{proof}
First, note that $r_1$ exists because $\mathbb{X}$ is convex and non-empty and has the origin in its interior, and $r_2$ exists because $O_{n-1}$ is compact.  Second, note that $\Omega_{r_1} \subset O_\infty$ because  $\Omega_{r_1}$ is an invariant, constraint-admissible set and $O_\infty$ contains all such sets (see \cite{Gilbert_1991}). Thus, we have the following inclusions:  $\Omega_{r_1} \subset O_\infty \subset O_{n-1} \subset \Omega_{r_2}$, which means that $r_2 \geq r_1$, as required. 

The rest of the proof leverages two facts from linear algebra. First, the eigenvalues of a symmetric, positive-definite matrix are all real and positive. Second, for any $P=P^{\sf T}\succ 0$, we have that $\lambda_{min}(P) x^Tx \leq x^TPx\leq \lambda_{max}(P)x^Tx$, where $\lambda_{min}$ and $\lambda_{max}$ are well-defined thanks to the first fact. Given $V(x)$ in \eqref{eq:lyapunov}, the second fact allows us to write 
$$
-x^Tx \leq -\frac{V(x)}{\lambda_{max}(P)}
$$
which we use below. We now write the change in the Lyapunov function along the trajectories as:
$$
V(x(t+1))-V(x(t)) = (Ax(t))^T P (Ax(t)) - x(t)^TPx(t)  
$$
$$
= x(t)^T (A^TPA-P) x(t) = - x(t)^TQx(t)
$$
$$
\leq -\lambda_{min}(Q)x(t)^Tx(t) \leq -\frac{\lambda_{min}(Q)}{\lambda_{max}(P)} V(x(t))
$$
The above can be rewritten as: 
\begin{equation}\label{eq:lyapunov decay}
V(x(t+1)) \leq \sigma V(x(t)) \Rightarrow V(x(t)) \leq \sigma^t V(x(0)) 
\end{equation}
where $\sigma$ is as defined in the Theorem. 

One way to quantify an upper bound on $t^*$ is to find the longest time it takes for any initial state within $O_{n-1}$ to enter $\Omega_{r_1}$. Indeed, if $x(t) \in \Omega_{r_1}$, then $x(t) \in \mathbb{X}$ and thus $y(t) \in \mathbb{Y}$ for all future times due to the invariance of $\Omega_{r_1}$. However, instead of $O_{n-1}$, we consider initial states within $\Omega_{r_2}\supset O_{n-1}$, which allows for simple computations using the ellipsoidal mathematics at the expense of making the upper bound less tight. To this end, note that any $x(0) \in \Omega_{r_2}$ satisfies $V(x(0))\leq r_2$. Therefore, $V(x(t)) \leq \sigma^t r_2$. Furthermore, to ensure $x(t) \in \Omega_{r_1}$, we must have $V(x(t))\leq r_1$. Therefore, we set $
V(x(t)) \leq \sigma^t r_2 \leq r_1,
$
which implies that
$$
t > 
\frac{\log (\frac{r_1}{r_2})}{\log (\sigma)}
$$
Any $t$ satisfying the above will be an upper bound on  $t^*$. Since we are interested in the largest integer time-step after which the constraints are redundant, we take the floor of the right hand side of the above, which completes the proof.
\end{proof}

Procedures for computing $r_1$ and $r_2$ in the theorem are well-established, see, e.g., \cite{boyd2004convex}. Specifically, $r_1$ can be found by
\begin{equation}\label{eq:k1}
r_1 = \min_j \frac{(\min\{y_j^l, y_j^u\})^2}{c_j P^{-1} c_j^T}
\end{equation}
To find $r_2$, $O_{n-1}$ can first be converted from the H-representation (i.e., half-space description) to V-representation (i.e., vertex description) \cite{avis2002canonical}. Let the vertices of $O_{n-1}$ in the V-representation be denoted by $v_j$. Then, $r_2$ can be found by 
\begin{equation}\label{eq:k2}
r_2 = \max_j \{v_j^TPv_j\}
\end{equation}

\begin{remark}
Polynomial time algorithms exist that can convert a polytope from the H-representation to V-representation \cite{avis2002canonical}. However, these algorithms may be computationally intensive, e.g., in higher dimensions. This may hamper the use of Method~2, e.g., in situations where the upper bound on $t^*$ must be computed online. To remedy this, one can take one of the following three approaches, all of which result in a larger (less tight) upper bound on $t^*$: (1) replace $O_{n-1}$ by any compact superset with known vertices, if such superset is available. (2) Apply the algorithms in \cite{makarychev2022streaming} to directly find the bounding ellipsoid $\Omega_{r_2}$. These algorithms are applicable in situations where information about the size or aspect ratio of $O_{n-1}$ is available. (3) Starting from a ``template" polytope, $G$, with known H- and V-representations, find a scaling $\alpha$ such that $O_{n-1} \subset \alpha G$; this scaling can be found by solving efficient linear programs. Since the vertices of $G$ are known, the vertices of $\alpha G$ are also known, so $\Omega_{r_2}$ can be found using \eqref{eq:k2} by replacing $v_j$ with the vertices of $\alpha G$. A comparison of the scalability and computational aspects of these methods is an interesting topic for future research. 
\end{remark}

It remains to find a suitable $Q$ to solve for $P$ using \eqref{eq:lyap Q}. We approach this problem by analytically finding $Q$ that results in the smallest $\sigma$ in Theorem \ref{theorem:lyapunov bound} and thus the fastest decay rate of the Lyapunov function along the system trajectories (see Eq.~\eqref{eq:lyapunov decay}). Note that this is not necessarily the $Q$ that results in the globally minimal value for the upper bound on $t^*$. Other possible approaches for selecting $Q$ include solving an optimization problem to find a $Q$ that minimizes the upper bound; finding a $Q$ such that $\Omega_{r_1}$ has the largest volume; or finding the $Q$ such that $\Omega_{r_2}$ has the smallest volume. These approaches, however, require nonlinear program or second-order cone program solvers, which is what we seek to avoid. Furthermore, our numerical studies showed that choosing $Q$ to minimize $\sigma$ led to the best possible upper bound in most situations. We will illustrate this in Section \ref{sec:numerical comparison unforced}.

\begin{theorem}\label{thm:sigma}
The scalar $\sigma$ in Theorem \ref{theorem:lyapunov bound} satisfies $0 \leq \sigma < 1$. Furthermore, the matrix $Q$ that results in the smallest $\sigma$ is $Q=I$, and the corresponding value of $\sigma$ is $\sigma = \rho(A)^2$, where $\rho:=\max_i|\lambda_i(A)|$ is the spectral radius of $A$.
\end{theorem}
\begin{proof}
To prove the first part, note that $V(x(t+1)) \geq 0$ and $V(x(t))\geq 0$ in Eq.~\eqref{eq:lyapunov decay}. Thus, $\sigma \geq 0$. To show $\sigma < 1$, note that $\lambda_{min}(Q), \lambda_{max}(P) > 0$. Thus, $\frac{\lambda_{min}(Q)}{\lambda_{max}(P)} > 0$, which implies that $\sigma = 1-\frac{\lambda_{min}(Q)}{\lambda_{max}(P)} < 1$.

To prove the second part, we must find $Q$ to maximize $\frac{\lambda_{min}(Q)}{\lambda_{max}(P)}$. To begin, 
note that by linearity of the Lyapunov equation in \eqref{eq:lyap Q}, normalizing $Q$ by any scalar normalizes $P$ by the same scalar. Therefore, without loss of generality, one can normalize $Q$ such that $\lambda_{min}(Q)=1$, which implies that $Q \succeq I$ or $Q-I\succeq 0$. Since $\lambda_{min}(Q)=1$, it now suffices to find a $Q$ to minimize $\lambda_{max}(P)$.

It is known that the solution, $P$, of the Lyapunov equation \eqref{eq:lyap Q} can be expressed as \cite{Chen_1998}:
\begin{equation}\label{eq:powerseriesLyap}
P(Q) = \sum_{t=0}^{\infty} (A^T)^t Q A^t
\end{equation}
We can thus write:
$$
P(Q)-P(I) = \sum_{t=0}^{\infty} (A^T)^t (Q-I) A^t 
$$
Since $Q-I \succeq 0$, we have that $P(Q)-P(I) \succeq 0$ or equivalently, $P(Q) \succeq P(I)$. Thus $\lambda_{max}(P(Q)) \geq \lambda_{max}(P(I))$ so to minimize the largest eigenvalue of $P$, one must take $Q=I$. 

 Finally, to show that the choice of $Q=I$ leads to $\sigma = \rho(A)^2$, we again leverage \eqref{eq:powerseriesLyap} and redefine $\bar{A} = A^T A$. We then apply the spectral mapping theorem from linear algebra to conclude that $\lambda_i (P) = \sum_{t=0}^{\infty} \lambda_i(\bar{A})^t = \frac{1}{1-\lambda_i(\bar{A})}$. Since $\lambda_i(\bar{A}) = (\lambda_i(A))^2$, we have that $\lambda_{max}(P) = \frac{1}{1-\rho(A)^2}$, which implies that $\sigma = \rho(A)^2$.
\end{proof} 

The above results  lead to Algorithm \ref{Algorithm:algorithm2} for finding an upper bound for $t^*$. 

\begin{algorithm}
 \caption{Compute upper bound on $t^*$ using Method 2}
     \hspace*{\algorithmicindent} \textbf{Input: $A, C, y_j^l, y_j^u$} \\
    \hspace*{\algorithmicindent} \textbf{Output: $m$ such that $t^*\leq m$} \\
 \begin{algorithmic}[1]
  \STATE Compute $P$ using \eqref{eq:lyap Q} with $Q=I$. Compute $\sigma = \rho(A)^2$
 \STATE Compute $r_1$ using \eqref{eq:k1}. 
 \STATE Construct $O_{n-1}$ as in \eqref{eq:On-1}, convert to V-representation, and compute $r_2$ using \eqref{eq:k2}. 
 \STATE  Compute $m$ using expression \eqref{equ:m_method2}.
 
 \end{algorithmic} 
 \label{Algorithm:algorithm2}
 \end{algorithm}

A numerical comparison between the two methods is provided in the next subsection.
 
\subsection{Numerical Comparison}\label{sec:numerical comparison unforced}

This section presents a comparative analysis of the upper bounds provided by Algorithm~\ref{Algorithm:algorithm1}
for Method~1 (i.e., the power series-based method) and Algorithm~\ref{Algorithm:algorithm2} for Method~2 (i.e., the Lyapunov-based method). Since this comparison cannot be carried out analytically, we conduct a Monte Carlo study of randomly-generated systems using Matlab 2020b.

To generate each random system, we first randomly generate $n$, the order of the system, by sampling the uniform distribution between 1 and 8. We then generate a state-space model with that order by using Matlab's {\tt drss} command, which returns Lyapunov stable systems with possibly repeated poles. To ensure Assumption \ref{assumption1} is robustly satisfied, we reject systems for which the spectral radius is greater than 0.999 and the smallest singular value of the observability matrix is less than 0.0001. For simplicity, we assume a single output (i.e., $q=1$) and symmetric constraints $y_1^u = y_1^l=1$. 

Using the above methodology, we generate a total of 16,000 random systems. We assume that the input satisfies $u=0$, which makes each system have the form \eqref{eq:unforced system}. For each system, we compute $t^*$ using the algorithm described in \cite{Gilbert_1991}. We also compute the upper bounds on $t^*$ using Algorithms \ref{Algorithm:algorithm1} and \ref{Algorithm:algorithm2}. We denote these upper bounds by $m_1$ and $m_2$ respectively, where the subscript refers to the respective method. To compare the upper bounds against the true value of $t^*$, we construct the histograms of $m_i -t^*$, $i=1,2$, as seen in Fig. \ref{fig:method1}. In addition to the histograms, a point by point comparison between the two methods is provided in Fig. \ref{fig:comparison}. As can be seen from the data, Method 1 performs well overall, with a median of 0 (i.e., for at least half of the random systems, the upper bound is tight). Furthermore, interestingly, Method 1 outperforms Method 2 \underline{in all cases}. Investigation of this observation is an interesting topic for future research.

\begin{figure}
\centering
\includegraphics[width=\columnwidth]{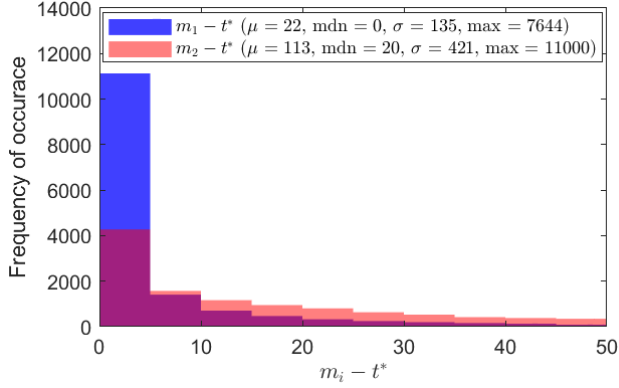}
\caption{Histograms of $m_1-t^*$ and $m_2-t^*$ (i.e., the tightness of each upper bound) obtained using our Monte Carlo study. In the legend, $\mu$, $\sigma$, and mdn refer to the mean, standard deviation, and median, respectively.}
\label{fig:method1}
\end{figure}


\begin{figure}
\centering
\includegraphics[width=0.8\columnwidth]{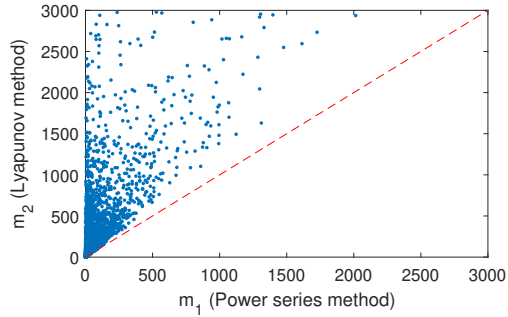}
\caption{Comparison between the upper bound provided by Methods 1, $m_1$, and by Method 2, $m_2$. Interestingly, $m_1 \leq m_2$ in all cases.}
\label{fig:comparison}
\end{figure}

From these figures, it may appear that the upper bounds are too conservative for some systems, which, per Remark~\ref{remark:bound}, could be attributed to the large spectral radius of those systems. This can be easily confirmed with our Monte Carlo study, as seen in Fig. \ref{fig:error vs rho} for Method 1. To investigate further, we normalize both $t^*$ and its upper bound $m_1$ to allow for a fair comparison between the different systems. The normalization is achieved by scaling $t^*$ and $m_1$ by $\log (\rho)$, where $\rho = \max_i (|\lambda_i(A)|)$ is the spectral radius. Taking  logarithms is inspired by the fact that continuous-time poles and discrete-time poles are related through $z=e^{sT_s}$, where $T_s$ is the sample time. Assuming $T_s=1$ to allow for direct comparison between the systems, we obtain $s=\log(z)$. Thus, scaling by $\log (\rho)$ normalizes each $t^*$ or $m$ by the ``continuous-time time constant" of the system. The results are reported in Fig. \ref{fig:spectral}. As can be seen, in the normalized coordinates, the spread is narrow and the upper bound is not as conservative as it appeared before. Similar plots can be generated for Method 2.

\begin{figure}
\centering
  \subfloat[$m_1-t^*$ vs. spectral radius, $\rho$. The larger the $\rho$, the more conservative the upper bound may be.\label{fig:error vs rho}]{%
\includegraphics[width=0.47\columnwidth]{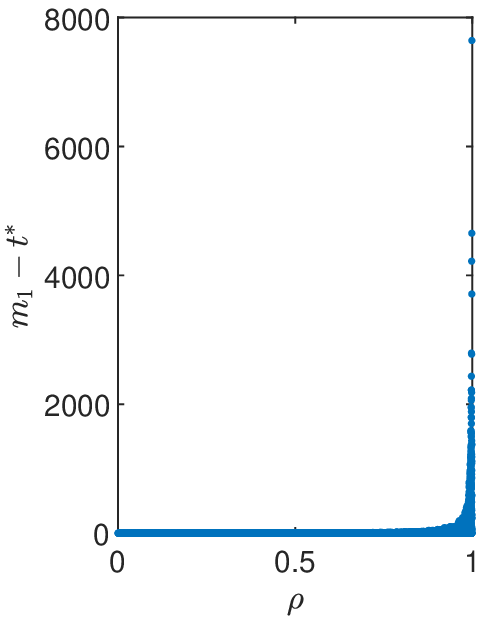}}
\hfill
  \subfloat[Comparison of $t^*$  vs. $m_1$, each scaled by the logarithm of the spectral radius of $A$.\label{fig:spectral}]{%
\includegraphics[width=0.47\columnwidth]{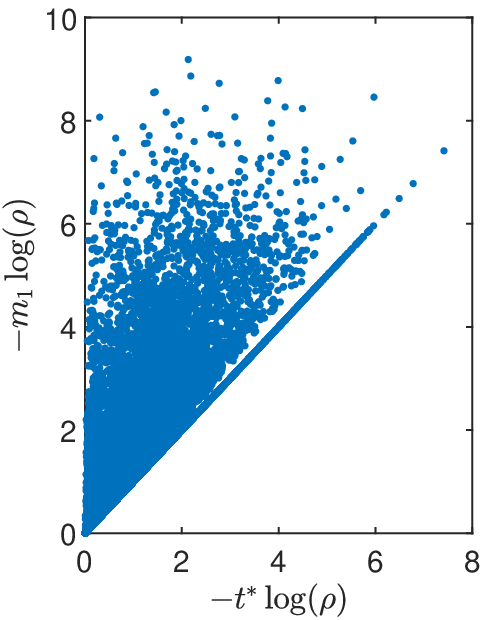}}

 \caption{Analysis of the upper bounds obtained using Method 1.}
 \label{fig:study of large bounds}
\end{figure}

Next, recall that in the Lyapunov-based approach of Method 2, we choose $Q=I$ and compute $P$ using Lyapunov equation \eqref{eq:lyap Q}, which results in the smallest possible value of $\sigma$, see Theorem~\ref{thm:sigma}. However, this choice of $Q$ may not necessarily lead to the smallest upper bound that can be obtained using the Lyapunov-based approach. To investigate the optimality of $Q=I$, we formulate the following nonlinear optimization problem, whose objective function is the upper bound on $t^*$ (see \eqref{equ:m_method2}) but without the floor operator:
$$
\min_{Q} \frac{\log (\frac{r_1}{r_2})}{\log( \sigma)}
$$
subject to $Q=Q^{\sf T} \succ 0$, and the following equality constraints: $r_1(P)$ and $r_2(P)$ from \eqref{eq:k1}--\eqref{eq:k2}, $\sigma(P,Q)$ from \eqref{eq:sigma}, and $P(Q)$ from \eqref{eq:lyap Q}. For each of the 16,000 random systems, we solve this problem, starting with the initial guess of $Q=I$, using Matlab's ``fmincon" function until a local minimum is reached. The optimal upper bound on $t^*$, which is what we seek to find, is  obtained by applying the floor operator to the  objective function value at the optimum. Based on the results, we make two interesting observations. First, the upper bound obtained using this optimization problem is still larger than that obtained using Method 1 for all random systems considered. Second, the upper bounds obtained using $Q=I$ and the one obtained using the $Q$ from the above optimization problem were identical for 15,764 (i.e. 98.5\%) of the systems, which provides additional justification for the efficacy of $Q=I$. The systems in which the two upper bounds differed had large spectral radii, leading to large values of $\sigma$ in \eqref{eq:sigma} and thus large sensitivity of the objective function to problem data.  


 In the above Monte Carlo study, symmetric constraints were assumed. We conclude this section with a numerical example to illustrate the effects of asymmetry on $t^*$ and its upper bounds $m_1$ and $m_2$. Consider system \eqref{eq:unforced system} with one output and the following system matrices:
$$
A=\left[ \begin{array}{ccc} 0.9 & -0.25 & 1 \\ 0.25 & 0.9 & 0 \\ 0 & 0 & -0.98\end{array} \right],\,\,  C=[-1 \,\,\,\, 1\,\,\,\, 0.5]$$
We let $y_1^u=1$ and vary $y_1^l$ (the lower constraint) from 0.1 to 2. Within this range, the constraint set is symmetric for $y_1^l=1$ and asymmetric otherwise. For each $y_1^l$, we compute $t^*$ and its upper bounds using Algorithms~1 and 2. The results are shown in Fig.  
\ref{fig:asymmetrystudy}. As can be seen, asymmetry tends to increase $t^*$ and its upper bounds, which aligns with the theoretical results in the previous section through Eqs. \eqref{eq:gamma} and \eqref{eq:k1}. This figure also illustrates that, similar to the symmetric case examined before, $m_1$ is a tighter bound than $m_2$ even in the asymmetric case.

\begin{figure}
\centering
\includegraphics[width=0.8\columnwidth]{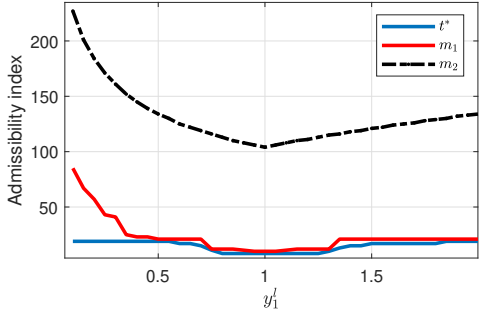}
\caption{Effect of asymmetry on $t^*$ and its upper bounds.}
\label{fig:asymmetrystudy}
\end{figure}

\section{Main Results: Systems with Constant Input}\label{sec:forced}

We now extend the results in the previous section to the forced system \eqref{eq:forced system} with constraint \eqref{eq:constraint}. As explained in Section~\ref{sec:intro}, the MAS for this system may not be finitely determined. However, by tightening the steady-state constraint, a finitely-determined inner approximation, denoted by  $\widetilde{O}_\infty$, can be obtained, see Eq. \eqref{eq:MAS forced}. For a given steady-state margin $\epsilon>0$, our goal is to obtain upper bounds on $t^*$ such that all constraints after time-step $t^*$ are guaranteed to be redundant in \eqref{eq:MAS forced}. Similar to the unforced case, we assume that:
\begin{assumption}
System \eqref{eq:forced system} is asymptotically stable and the pair $(A,C)$ is observable. It is assumed that the input $u$ is constant for all time. Furthermore, the constraint set in \eqref{eq:constraint} is described by \eqref{eq:constraint detailed}.
\label{assumption2}
\end{assumption}

\subsection{Method 1: Matrix Power Series}

Similar to the case of unforced systems, the general idea behind this method is finding an expansion of $A^t$, with ``sufficiently small" coefficients, in terms of lower powers of $A$. The key difference with the unforced case is that the origin is no longer the equilibrium of the forced system, so we must perform a change of coordinates to shift the equilibrium to the origin. Furthermore, recall from \eqref{eq:MAS forced} that the steady-state constraint is tightened by $(1-\epsilon)$, which introduces additional complexities. 

Under Assumption \ref{assumption2}, the equilibrium of \eqref{eq:forced system} is given by 
$$
x(\infty) = (I-A)^{-1}Bu, \quad y(\infty) = H_0 u
$$
where
$$
H_0 = C(I-A)^{-1}B+D
$$
is the DC gain from $u$ to $y$. Note that the matrix inverse exists thanks to the asymptotic stability of $A$. We define a new state vector $z(t)$ to shift the equilibrium to the origin:
$$
z(t) = x(t) - (I-A)^{-1}Bu
$$
In the new coordinate system, the dynamics are described by:
\begin{equation}\label{eq:forced system new coordinates}
\begin{aligned}
z(t+1) = Az(t)\\
y(t) = Cz(t) + H_0 u
\end{aligned}
\end{equation}
The output thus evolves according to 
$$
y(t) = CA^tz(0) + H_0 u
$$
We now state the main result of this section.

\begin{theorem}\label{theorem:forced}
Consider system \eqref{eq:forced system} with constraint \eqref{eq:constraint}, and suppose 
Assumption \ref{assumption2} holds. Suppose there exists an integer $m$, $m \geq 0$, such that $A^{m+1}$ can be expanded as in \eqref{eq:expansion}, 
where $\alpha_i$ satisfy:
\begin{equation}\label{eq:condition on alphai forced}
\Big(1+\gamma(1-\epsilon)\Big) \sum_{\alpha_i > 0} \alpha_i - \Big(\gamma +(1-\epsilon)\Big)
\sum_{\alpha_i < 0} \alpha_i  \leq\epsilon
\end{equation}
and $\gamma$ is defined in \eqref{eq:gamma}.
Then, $t^* \leq m$.
\end{theorem}
\begin{proof}
The proof is similar to that of Theorem \ref{theorem:unforced} with some differences, which we highlight. As in Theorem 1, we use mathematical induction to prove that $y(t)\in \mathbb{Y}$ for $t \leq m$ implies that $y(t)\in \mathbb{Y}$ for $t \geq m+1$. For the sake of brevity, we only discuss the base case of the induction argument, as the proof of the induction step is similar.

For the induction base case,  we assume that $y(t)\in \mathbb{Y}$ for $t \leq m$ and show that $y(m+1) \in \mathbb{Y}$. To show this, note that the $j$-th output can be written as:
$$
y_j(m+1) = C_j A^{m+1} z_0 + H_0 u = \sum\limits_{i=0}^{m} \alpha_i (C_j A^{i} x_0) + H_0u
$$
We add and subtract $\sum_{i=0}^{m} \alpha_i H_0 u$ to this expression to obtain:
$$
y_j(m+1) = \sum\limits_{i=0}^{m} \alpha_i (C_j A^{i} x_0+H_0u) + H_0u-\sum_{i=0}^{m} \alpha_i H_0 u
$$
The assumption $y(t) \in \mathbb{Y}$ for $t\leq m$ implies that $C_j A^ix_0 + H_0 u$ in the first sum satisfies: $-y_j^l \leq C_j A^ix_0 + H_0u\leq y_j^u$. Furthermore, the assumption $y(\infty) \in (1-\epsilon) \mathbb{Y}$ implies that     $-(1-\epsilon) y_j^l \leq H_0 u \leq (1-\epsilon) y_j^u$. Thus, breaking up the sum into positive and negative values of $\alpha_i$ as we did in the proof of Theorem \ref{theorem:unforced},  we obtain the following bounds on $y_j(m+1)$:
 $$
 -y_j^l (1-\epsilon) - y_j^l  \sum_{\alpha_i>0}\alpha_i + y_j^u \sum_{\alpha_i<0}\alpha_i - y_j^u(1-\epsilon) \sum_{\alpha_i>0}\alpha_i +
 $$
 $$
 y_j^l(1-\epsilon) \sum_{\alpha_i<0}\alpha_i \leq  \boxed{y_j(m+1)} \leq  y_j^u (1-\epsilon) + y_j^u  \sum_{\alpha_i>0}\alpha_i - 
 $$
 $$
y_j^l \sum_{\alpha_i<0}\alpha_i + y_j^l(1-\epsilon) \sum_{\alpha_i>0}\alpha_i - y_j^u(1-\epsilon) \sum_{\alpha_i<0}\alpha_i 
 $$
 To ensure that $-y_j^l \leq y(m+1)\leq y_j^u$, we set the left inequality to be greater than $-y_j^l$ and the right inequality to be smaller than $y_j^u$. We then divide the left inequality by $-y_j^l$ and the right inequality by $y_j^u$ and simplify terms to obtain:
$$
   \left(1+\frac{y_j^l}{y_j^u}(1-\epsilon)\right) \sum_{\alpha_i>0}\alpha_i - \left(\frac{y_j^l}{y_j^u}+(1-\epsilon)\right)\sum_{\alpha_i<0}\alpha_i \leq \epsilon
$$
$$
   \left(1+\frac{y_j^u}{y_j^l}(1-\epsilon)\right) \sum_{\alpha_i>0}\alpha_i - \left(\frac{y_j^u}{y_j^l}+(1-\epsilon)\right)\sum_{\alpha_i<0}\alpha_i \leq \epsilon
$$
Both of these inequalities hold since they are implied by \eqref{eq:condition on alphai forced}. Thus, $y(m+1)\in\mathbb{Y}$.
\end{proof}

\begin{remark}\label{remark:forced symmetric}
In the case of symmetric constraints, the expression in Theorem \ref{theorem:forced} can be further simplified. Specifically, suppose  that $y_j^l = y_j^u, \forall j$ in \eqref{eq:constraint detailed}. Then, $\gamma = 1$ and so condition \eqref{eq:condition on alphai forced} becomes:
\begin{equation}\label{eq:forced condition symmetric}
\sum_{i} |\alpha_i| \leq \frac{\epsilon}{2-\epsilon}
\end{equation}
Note that the right hand side tends to 0 and 1, as $\epsilon$ tends to 0 and 1, respectively. 
\end{remark}

As in the unforced case, the Cayley-Hamilton based expansion of Section \ref{sec:unforced} can be employed to obtain the expansion in Theorem \ref{theorem:forced} and thus obtain an upper bound on $t^*$. An algorithm similar to Algorithm \ref{Algorithm:algorithm1} can be constructed for this purpose, wherein condition \eqref{eq:condition on betai} is replaced with:
\begin{equation}\label{eq:condition on betai forced}
\Big(1+\gamma(1-\epsilon)\Big) \sum_{\beta_i(t) > 0} \beta_i(t) - \Big(\gamma +(1-\epsilon)\Big)\sum_{\beta_i(t) < 0} \beta_i(t)  \leq\epsilon
\end{equation}
The complete algorithm is provided in Algorithm \ref{Algorithm:algorithm3}.
\begin{algorithm}
 \caption{Compute upper bound on $t^*$ using Method 1 for the case of systems with constant input.}
     \hspace*{\algorithmicindent} \textbf{Input: $A, y_j^l, y_j^u, \epsilon$} \\
    \hspace*{\algorithmicindent} \textbf{Output: $m$ such that $t^*\leq m$} \\
 \begin{algorithmic}[1]
 \STATE Compute the Cayley Hamilton coefficients, $c_i$, using \eqref{eq:characteristic polynomial}, and $\gamma$ using \eqref{eq:gamma}.
 \STATE Set $t=n$ and initialize $\beta(n)$ as in Theorem \ref{lemma:CH coeff}.
 \STATE  If $\beta(t)$ satisfies \eqref{eq:condition on betai forced}, then: $m=t-1$, STOP.
 \STATE Increment $t$ by 1. Compute $\beta(t)$ using \eqref{eq:CH recursion}. Go to step 3.
 \end{algorithmic} 
 \label{Algorithm:algorithm3}
 \end{algorithm}

Similar to the unforced case, the above results allow us to simplify the computation of MAS for some first order systems:
\begin{theorem}
Consider \eqref{eq:forced system} with $A\in\mathbb{R}$ (i.e., a first order system) and assume Assumption \ref{assumption2} holds. If 
\begin{equation}\label{eq:firstorderforced}
\frac{-\epsilon}{\gamma+(1-\epsilon)}\leq A<\frac{\epsilon}{1+\gamma(1-\epsilon)}
\end{equation}
then $t^*=0$. 
\end{theorem}
\begin{proof}
From \eqref{eq:CH n}, it follows that $A=-c_0$ and from Theorem~\ref{lemma:CH coeff}, $\beta(t) \in \mathbb{R}$ satisfies $\beta(1)=-c_0=A$. Thus, $\beta(1)$ also satisfies condition \eqref{eq:firstorderforced}. Using this condition, and the fact that $0< \frac{\epsilon}{1+\gamma(1-\epsilon)}<1$ and $-1 < \frac{-\epsilon}{\gamma+(1-\epsilon)}<0$, it can be seen that, regardless of sign of $\beta(1)$,  \eqref{eq:condition on betai forced} is satisfied for $t=1$. Per Algorithm~\ref{Algorithm:algorithm3}, an upper bound on $t^*$ is therefore $m=0$, which implies that $t^*=0$.
\end{proof}

We conclude this section with two remarks.

\begin{remark}\label{remark:epsilon}
In condition \eqref{eq:condition on betai forced}, the smaller the $\epsilon$ (i.e., the steady-state tightening), the smaller the right hand side, and therefore the smaller the $\beta_i(t)$ must be to satisfy the condition. According to Theorem \ref{lemma:CH coeff}, smaller $\beta_i(t)$'s are achieved with larger $t$'s. Therefore, the upper bound on $t^*$ (and likely $t^*$ itself) grows as $\epsilon$ becomes small. Furthermore, if $\epsilon \ll 1$ (which is typical in applications),  \eqref{eq:condition on betai forced} can be approximated by 
$$
\sum_{i} |\beta_i(t)| \leq \frac{\epsilon}{1+\gamma}
$$
which implies that $|\beta_i(t)| \ll 1$. Thus, for the same constraints $y_j^l$, $y_j^u$, and the same matrices $A$ and $C$, $\beta_i(t)$'s that satisfy this condition are likely smaller than those that satisfy \eqref{eq:condition on betai}. Thus, the upper bound on $t^*$ (and likely $t^*$ itself) is larger in the forced case than the unforced case.
\end{remark}

\begin{remark}

Note that condition \eqref{eq:condition on betai} for the unforced case and \eqref{eq:condition on betai forced} for the forced case become identical when $\epsilon = 1$ (this forces $u=0$ due to $H_0u\in (1-\epsilon)\mathbb{Y}$, which makes intuitive sense). In this sense, Method 1 in the forced case can be viewed a proper extension of Method 1 in the unforced case.
  \end{remark}

\subsection{Method 2: Lyapunov Level Sets}

The second method, which relies on Lyapunov level sets to find an upper bound on $t^*$, requires only minor modifications compared to the input-free case.  We first extend the definition of $O_{n-1}$ in \eqref{eq:On-1} to account for the input in system
\eqref{eq:forced system new coordinates}, where we tighten the steady-state constraint similar to \eqref{eq:MAS forced}:
\begin{align} \label{eq:ivk1}
\widetilde{O}_{n-1} = \big\{(z_0,u): 
 & ~C A^t z_0+H_0 u \in \mathbb{Y}, t=0,\ldots,n-1, \nonumber \\
& H_0 u \in (1-\epsilon) \mathbb{Y}
\big\}. 
\end{align}
As in the case of $O_{n-1}$, this set is a compact polytope. We have the following result.
\begin{theorem}\label{theorem:forced2}
Consider system \eqref{eq:forced system} with Lyapunov function \eqref{eq:lyapunov}--\eqref{eq:lyap Q} and constraint \eqref{eq:constraint}, and suppose Assumption \ref{assumption2} holds. Define $r_1, r_2 \in \mathbb{R}$ as follows: 
\begin{equation}\label{eq:r1forced}
r_1 = \max \big\{r:~\Omega_r \subset \epsilon \mathbb{X} \big\}
\end{equation}
\begin{equation}\label{eq:r2forced}
r_2 = \min \big\{r:~{\tt Proj}_{z} \widetilde{O}_{n-1} \subset \Omega_r \big\},
\end{equation}
where ${\tt Proj}_{z}$ denotes the projection onto $z$-coordinates.
Then, an upper bound on $t^*$ is given by expression \eqref{equ:m_method2}. 
\end{theorem}
\begin{proof}
Suppose $(z(0),u) \in \widetilde{O}_{n-1}$. Then, by \eqref{eq:ivk1} and \eqref{eq:r2forced}, $z(0) \in \Omega_{r_2}$ and $H_0 u \in (1-\epsilon) \mathbb{Y}$.
By the same arguments as in the proof of Theorem~\ref{theorem:lyapunov bound} applied to \eqref{eq:forced system new coordinates}, $z(t)$ starting from such $z(0)$ satisfies $z(t) \in \Omega_{r_1}$ for all $t \geq m$,
where $m$ is given by expression \eqref{equ:m_method2} with $r_1$ and $r_2$ given by \eqref{eq:r1forced}--\eqref{eq:r2forced}. This, together with \eqref{eq:r1forced}, implies that $z(t) \in \epsilon \mathbb{X}$ or, equivalently, $Cz(t) \in \epsilon \mathbb{Y}$ for all $t \geq m$, which implies that $y(t)$ in \eqref{eq:forced system new coordinates} satisfies:
$$
y(t) = Cz(t)+H_0u \in \epsilon\mathbb{Y}\oplus (1-\epsilon)\mathbb{Y} = \mathbb{Y}
$$
where $\oplus$ denotes the Minkowski set addition.
To summarize, the first $n$ set inclusions,
$CA^t z(0)+H_0 u \in \mathbb{Y}$, $t=0,  \ldots, n-1$, coupled with $H_0 u \in (1-\epsilon)\mathbb{Y}$, make redundant (i.e., automatically holding) the inequalities corresponding to $y(t)\in \mathbb{Y}$ for $t \geq m$. Thus $t^* \leq m$. 
\end{proof}

Procedures for computing $r_1$ and $r_2$ in the theorem are similar to those in Section \ref{sec:unforced}. Specifically, given $\epsilon \in (0,1)$, $r_1$ can be found by
\begin{equation}\label{eq:k1forcedeq}
r_1 = \min_j \frac{(\min\{\epsilon y_j^l, \epsilon y_j^u\})^2}{c_j P^{-1} c_j^T}
\end{equation}
To find $r_2$, 
we first compute 
 $\widetilde{O}_{n-1}$ using \eqref{eq:ivk1} 
and convert it into the V-representation.  Let the vertices of $\widetilde{O}_{n-1}$ in the V-representation be denoted by $v_j \in \mathbb{R}^{n+m}$, where the first $n$ components correspond to the $z$-coordinates and the next $m$ components correspond to the $u$-coordinates. 
Then, $r_2$ can be found by 
\begin{equation}\label{eq:k2forcedeq}
r_2 = \max_j \{\bar{v}_j^TP\bar{v}_j\}
\end{equation}
where $\bar{v}_j \in \mathbb{R}^n$ is a vector consisting of the first $n$ components of $v_j$.

 An algorithm similar to Algorithm \ref{Algorithm:algorithm2} can be constructed to find the upper bound using Theorem \eqref{theorem:forced2}, see Algorithm \ref{Algorithm:algorithm4}.

 \begin{algorithm}
 \caption{Compute upper bound on $t^*$ using Method 2 for the case of system with constant input}
     \hspace*{\algorithmicindent} \textbf{Input: $A, C, y_j^l, y_j^u, \epsilon$} \\
    \hspace*{\algorithmicindent} \textbf{Output: $m$ such that $t^*\leq m$} \\
 \begin{algorithmic}[1]
  \STATE Compute $P$ using \eqref{eq:lyap Q} with $Q=I$. Compute $\sigma = \rho(A)^2$
 \STATE Compute $r_1$ using \eqref{eq:k1forcedeq}. 
 \STATE Construct $\widetilde{O}_{n-1}$ as in \eqref{eq:ivk1}, convert to V-representation, and compute $r_2$ using \eqref{eq:k2forcedeq}. 
 \STATE  Compute $m$ using expression  \eqref{equ:m_method2}.
 
 \end{algorithmic} 
 \label{Algorithm:algorithm4}
 \end{algorithm}

\begin{remark}
Note that $r_1$ and $r_2$ in \eqref{eq:r1forced}--\eqref{eq:r2forced} are smaller than those defined in Theorem \ref{theorem:lyapunov bound} because $\epsilon < 1$. This means that the upper bound computed using Algorithm \ref{Algorithm:algorithm4} is generally larger than that computed using Algorithm \ref{Algorithm:algorithm2} for the case of unforced systems. Furthermore, note that if $\epsilon = 1$, then the unforced case and the forced case become identical. This makes intuitive sense because if $\epsilon=1$, then $u=0$ to ensure $H_0 u\in (1-\epsilon)\mathbb{Y}$. In this sense, Method 2 in the forced case can be seen as the proper extension of Method 2 in the unforced case.
\end{remark}

\subsection{Numerical Comparison}

In this section, we perform a Monte Carlo study similar to the one presented in Section \ref{sec:numerical comparison unforced} to compare the upper bound obtained using Method 1 (Algorithm \ref{Algorithm:algorithm3}) with that obtained using Method 2 (Algorithm \ref{Algorithm:algorithm4}). For this purpose, we choose $\epsilon = 0.01$, and use the same 16,000 random systems described in Section \ref{sec:numerical comparison unforced} but this time allow $u\neq 0$. For each system, we compute $t^*$ using the algorithm described in \cite{Gilbert_1991}. Comparing the value of $t^*$ in the forced case with that in the unforced case (Section \ref{sec:numerical comparison unforced}), we see that $t^*$ in the forced case is larger than the $t^*$ in the unforced case for \underline{all} the random systems considered, which is an interesting observation. 

For each of the 16,000 system, we also compute the upper bounds on $t^*$ using Algorithms \ref{Algorithm:algorithm3} and \ref{Algorithm:algorithm4}. We denote these upper bounds by $m_1$ and $m_2$ respectively, where the subscript refers to the respective method. Similar to the true value of $t^*$, we find that the upper bounds in the forced case are always larger than those in the unforced case (presented in Section \ref{sec:numerical comparison unforced}). This is consistent with Remark~\ref{remark:epsilon}.

To compare the upper bounds against the true value of $t^*$ in the forced case, we construct the histograms of $m_i -t^*$, $i=1, 2$, shown in Fig. \ref{fig:method1forced}. As seen from the histograms and the underlying data, Method 1 performs well overall, with a median of 0, and more importantly, it outperforms Method 2 in \underline{all} the random systems considered. 

\begin{figure}
\centering
\includegraphics[width=\columnwidth]{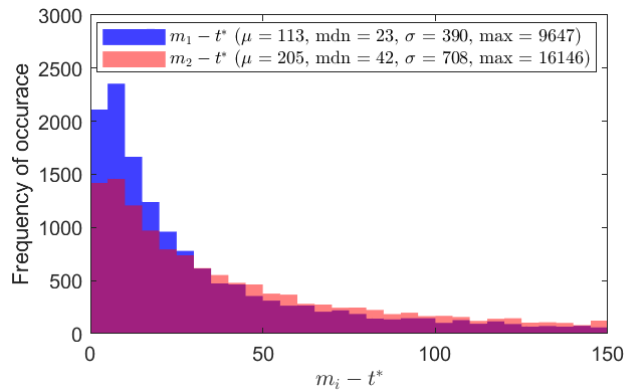}
\caption{Histograms of $m_1-t^*$ and $m_2-t^*$ (i.e., the tightness of each upper bound) from the Monte Carlo study. In the legend, $\mu$, $\sigma$, and mdn refer to the mean, standard deviation, and median, respectively.}
\label{fig:method1forced}
\end{figure}

\section{Conclusions and Future Work}\label{sec:conclusions}

This paper presented two computationally efficient methods to obtain {\it upper bounds} on the admissibility index of Maximal Admissible Sets for discrete-time LTI systems. The first method is algebraic and is based on matrix power series, while the second is geometric and is based on Lyapunov level sets. The two methods were rigorously introduced, a detailed numerical comparison between the two was provided, and the methods were extended to systems with constant inputs. It was shown that Method 1 outperforms Method 2, and that the upper bounds (and likely the admissibility index itself) depend on the spectral radius of matrix $A$ and also the steady-state tightening, $\epsilon$, in the case of systems with constant inputs. 

Future work will investigate the reason why Method 1 outperformed Method 2 in our numerical study. Another topic for future research is to find other power series expansions (beyond what is provided by the Cayley Hamilton method) to further improve the upper bounds in Method 1. Upper bounds for the admissibility index of {\it robust} maximal admissible set for systems with disturbances is another avenue of future research. 

\ifCLASSOPTIONcaptionsoff
  \newpage
\fi

\bibliographystyle{unsrt}
\bibliography{References}

\begin{IEEEbiography} [{\includegraphics[width=1in,height=1.25in,clip,keepaspectratio]{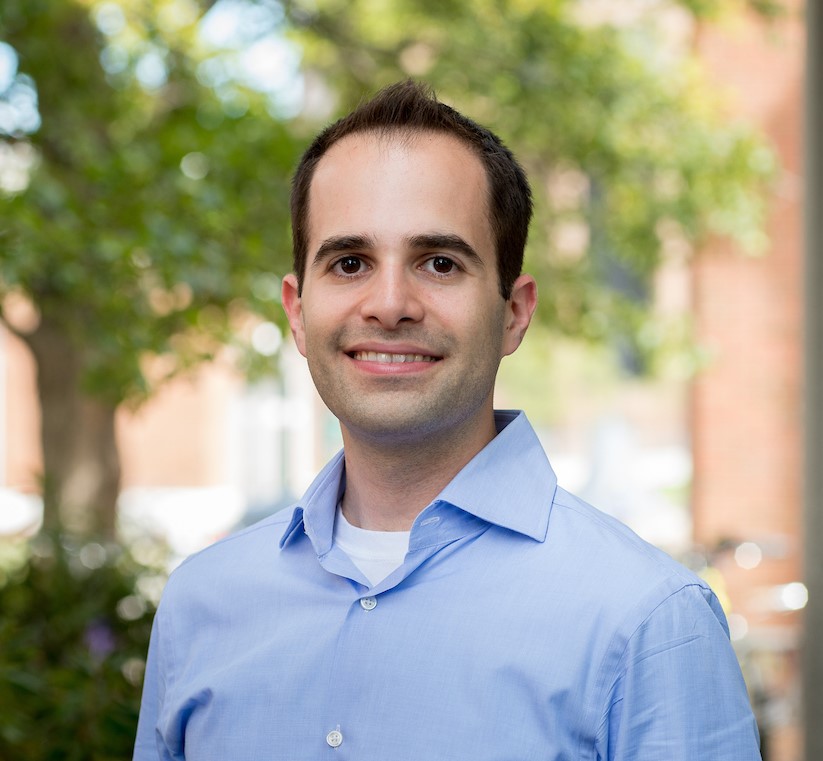}}]{Hamid R. Ossareh}
 (Senior Member, IEEE) received the
BA.Sc. degree from the University of Toronto in 2008,
and the Ph.D. degree from the University of Michigan, Ann Arbor in 2013. From 2013–2016, he was
with Ford Research and Advanced Engineering as a
Research Engineer. Since 2016, he has been a faculty member at the University of Vermont (UVM), currently at the rank of Associate Professor. His primary research
interests include the areas of systems and control theory, more specifically
predictive control, nonlinear control, and constrained control, with application areas of  automotive, power, aerospace, and xerographic systems. He holds
several patents and has won several awards, including the Faculty of the Year
award from IEEE Green Mountain Section, Excellence in Research Award from UVM College of Engineering and Mathematical Sciences, and Ford Technical Achievement Award from Ford Motor Company.
\end{IEEEbiography}

\begin{IEEEbiography} [{\includegraphics[width=1in,height=1.25in,clip,keepaspectratio]{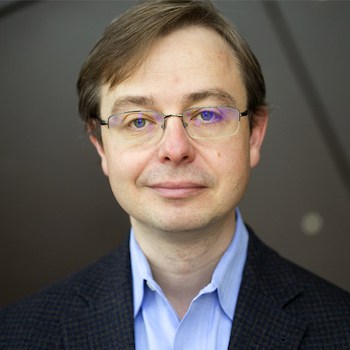}}]{Ilya Kolmanovsky }
(Fellow, IEEE) received the
Ph.D. degree in aerospace engineering from the
University of Michigan in 1995. He is currently a
Professor with the Department of Aerospace Engineering, University of Michigan, Ann Arbor, MI,
USA. Prior to joining the University of Michigan,
as a Faculty Member in 2010, he was with Ford
Research and Advanced Engineering in Dearborn,
Michigan, for close to 15 years. His research interests are in control theory for systems with state and
control constraints, and in control applications to
aerospace and automotive systems. He is a Senior Editor of IEEE Transactions on Control Systems Technology.
\end{IEEEbiography}


\end{document}